\documentclass[prd,nofootinbib,amsfonts,notitlepage,aps,10pt]{revtex4-2}
\usepackage[utf8]{inputenc}
\usepackage[T1]{fontenc}
\usepackage{hyperref}
\usepackage{graphicx,bm,amsmath,color}
\usepackage{bbold}
\usepackage{amssymb}
\usepackage{tikz}
\usepackage{graphicx}
\usepackage{float}
\usepackage{mathtools}
\usepackage{comment, cancel}
\makeatletter

\newcommand{\be}{\begin{equation}}
\newcommand{\ee}{\end{equation}}
\newcommand{\beq}{\begin{eqnarray}}
\newcommand{\eeq}{\end{eqnarray}}
\newcommand{\ba}{\begin{align}}
\newcommand{\ea}{\end{align}}
\newcommand{\red}[1]{{\leavevmode\color{red}{#1}}}

\newtheorem{theorem}{Theorem}[section]

\newtheorem{proposition}[theorem]{Proposition}
\newtheorem{corollary}[theorem]{Corollary}
\newtheorem{definition}[theorem]{Definition}

\newenvironment{proof}{\paragraph{Proof:}}{\hfill$\square$}

\begin{document}


\title{Generalized Hamilton spaces: new developments and applications}

\author{J.J. Relancio}
\affiliation{\small Departamento de Matemáticas y Computación, Universidad de Burgos, Plaza Misael Bañuelos, 09001, Burgos, Spain\\Centro de Astropartículas y F\'{\i}sica de Altas Energ\'{\i}as (CAPA),
Universidad de Zaragoza, C. de Pedro Cerbuna, 12, 50009 Zaragoza, Spain}
\email{jjrelancio@ubu.es}
\author{L. Santamaría-Sanz}
\affiliation{\small Departamento de Física, Universidad de Burgos, Plaza Misael Bañuelos, 09001 Burgos, Spain}
\email{lssanz@ubu.es}

\begin{abstract}
In this work, we make new developments in generic cotangent bundle geometries, depending on all phase-space variables. In particular, we will focus on the so-called generalized Hamilton spaces, discussing how the main ingredients of this geometrical framework, such as the Hamiltonian and the nonlinear and affine connections, can be derived from a given metric. Several properties of this kind of spaces are demonstrated for autoparallel Hamiltonians. Moreover, we study the spacetime and momentum isometries of the metric. Finally, we discuss the possible applications of cotangent bundle geometries in quantum gravity, such as the construction of deformed relativistic  kinematics and non-commutative spacetimes.
\end{abstract}

\maketitle


\section{Introduction}
(Co)tangent bundle geometries, which consider (momentum)velocity dependent metrics \cite{yano1973tangent}, have attracted renewed interest in the last years.  While the tangent bundle appears in Finsler \cite{Cartan1933} and Lagrange spaces~\cite{Kern1974, miron2001geometry}, the cotangent bundle structure can be examined within the so-called Hamilton and generalized Hamilton spaces~\cite{miron2001geometry}. The basic idea is to take a $n$-dimensional base manifold $M$ and construct the cotangent bundle manifold $T^*M$ formed by this base manifold (acting as the spacetime) and the fibers (the momentum space). Due to its multiple applications, these type of geometric structures are studied in several fields, such as Mathematics (see~\cite{miron2001geometry,Voicu:2021fim} and references therein), Medicine~\cite{WOS:000462628900019}, Biology~\cite{bookAntonelli}, 
and Physics~\cite{Pfeifer:2019wus}, to name a few. Some of their applications are the generalization of the study of Zermelo's navigation problem (to find the minimum time trajectory of a ship sailing on a sea with the presence of either constant~\cite{Zermelo1931berDN} and variable wind~\cite{aldea2017generalized}), the wildfire spreading~\cite{Markvorsen2016AFG}, fluid dynamics in inhomogeneous porous media~\cite{YAJIMA20151}, and seismic theory~\cite{ad8465ddcf864974a363481b760d60e1}, among many others. 

A more recent field of study, in which this kind of geometries has been recently considered, is quantum gravity (QG), a theory combining general relativity (GR) and quantum field theory. In particular, a possible way to effectively describe possible QG effects is by  modifying the kinematics of special relativity (SR). This can be done by preserving the relativity principle present in SR and GR, or by breaking it. The former case corresponds to doubly/deformed special relativity (DSR) theories~\cite{Amelino-Camelia2002b,AmelinoCamelia:2001vy,AmelinoCamelia:2008qg}, in which the deformation of the kinematics is present in the dispersion relation (Hamilton function) and the energy-momentum conservation law, while the latter approach is accomplished in Lorentz invariance violating (LIV) models~\cite{Colladay:1998fq,Kostelecky:2008ts}. Since the deformation of the kinematics is completely different in both approaches, they present very different phenomenologies~\cite{Addazi:2021xuf,AmelinoCamelia:2008qg,Amelino-Camelia:2011uwb,Carmona:2012un,AlvesBatista:2023wqm}, which suggest different ways of testing their physical predictions. 

Velocity and momentum dependent geometries have been considered as a way to describe modified Lagrange and Hamilton functions, leading in general to deformations of particle trajectories, in both LIV~\cite{Kostelecky:2011qz,Barcelo:2001cp,Weinfurtner:2006wt,Torri:2021hpj} and  DSR~\cite{Girelli:2006fw,AmelinoCamelia:2011bm,Amelino-Camelia:2014rga,Letizia:2016lew,Barcaroli:2015xda,Barcaroli:2016yrl,Lobo:2016blj,Barcaroli:2017gvg} scenarios. Regarding the latter, there is a clear way to connect the deformed kinematics of DSR and a curved momentum space~\cite{Carmona:2019fwf}:  both the Lorentz transformations and the composition law are isometries of a maximally symmetric momentum space, and the squared geodesic distance of that  metric can be identified with the Hamiltonian. Indeed, following the mathematical construction developed in~\cite{miron2001geometry} for Hamilton geometries, a series of papers~\cite{Relancio:2020zok,Relancio:2020rys,Relancio:2020mpa,Relancio:2021ahm,Relancio:2021asx,Pfeifer:2021tas,Relancio:2022mia,Relancio:2022kpf,Chirco:2022jvx,Mercati:2023pal} were devoted to the study of the cotangent bundle structure of a deformation of GR compatible with a deformed relativistic kinematics.  For that aim, metrics depending on both spacetime and momentum coordinates must be considered and, in particular, the ones corresponding to a generalized Hamilton space, because it gives the possible geometrical structure compatible with a relativistic deformation of the kinematics. 

In the literature, not so much attention has been put into generalized Hamilton spaces~\cite{miron2001geometry,WOS:000204390000017}. This could be in part due to the more complexity of these spaces (compared to Hamilton ones) and to the fact that not every ingredient needed for constructing its geometrical structure can be easily obtained once a metric is given. Indeed, as discussed in~\cite{miron2001geometry}, it is not yet obvious how to relate the nonlinear connection and the Hamiltonian to the metric. While the last point was clarified in~\cite{Relancio:2020rys}, showing that the squared distance in momentum space is a conserved quantity along horizontal curves (spacetime geodesics), the relation between the nonlinear connection and the metric is still missing.

In this paper we describe how, from a given metric, one is able to obtain the Hamiltonian and the affine connections. This differs from the usual construction of Hamilton spaces where, given a Hamiltonian, one is able to define the geometry~\cite{miron2001geometry}. Moreover, we demonstrate several properties of the affine connections and curvature tensors for autoparallel (conserved along horizontal paths) Hamiltonians, and in particular when they are either a homogeneous function on the momenta or an arbitrary function
of a homogeneous function on momenta. This property simplifies the construction of affine connections and curvature tensors.

It is also important to study the isometries of these generalized Hamilton spaces when considering a cotangent bundle metric depending on both spacetime and momentum coordinates.  These isometries can be studied in spacetime or in momentum space. The former case can be used to describe conserved quantities along horizontal paths (through Killing vector fields~\cite{Weinberg:1972kfs}). When requiring that the same number of spacetime isometries of the momentum independent metric is also present in the momentum dependent one, one is forced to restrict the possible momentum dependency of the metric, obtaining the same results of~\cite{Pfeifer:2021tas}. Moreover, we show that the squared momentum distance is indeed a Casimir of the algebra of isometries, explicitly showing examples in GR.  On the other hand, momentum isometries have been recently used for describing relativistic deformed kinematics~\cite{Carmona:2019fwf}. There are two main deformed kinematics studied in the literature, Snyder~\cite{Battisti:2010sr} and $\kappa$-Poincaré~\cite{Majid1994} kinematics.  Here, we show that the extension to curved spacetimes of both deformed kinematics can be obtained from a maximally symmetric momentum space when the metric depends on all the coordinates in phase space. We find that Snyder kinematics are privileged, since it is invariant under spacetime choices of observers (tetrads). Moreover, a possible connection between noncommutative spacetimes and the cotangent bundle geometry structure is proposed. Two possible definitions of noncommutativity of spacetime, from the momentum isometries generators and from the tetrad of the metrics, are considered. They also lead to the conclusion that Snyder noncommutativity seems to be privileged.

The structure of the paper is as follows. In Sec.~\ref{sec2} we give an overall view of  the necessary ingredients of the geometry of cotangent bundle geometry, and the basis of relativistic deformed kinematics in curved spacetime. In Sec.~\ref{sec3} we show a way of obtaining all the geometrical ingredients of a generalized Hamiltonian space once a generic metric is known, and demonstrate several properties of generalized Hamilton spaces. The discussion about isometries in spacetime and momentum space are described in Sec.~\ref{sec4} and ~\ref{sec5}, respectively. Finally, we collect the conclusions of this work in Sec.~\ref{conclusions}.

\section{Review of previous works}

\label{sec2}

\subsection{General notions of cotangent bundle geometries}

\begin{definition}[Cotangent bundle]
    The cotangent bundle $T^*M$ is a structure ($M, \pi, T^*_pM$) formed by the union of all the cotangent spaces $T^*_pM$ placed at each point $p$  of a smooth 4-dimensional manifold $M$. The bundle projection is given by $\pi:T^*M \to M;(x,k)\mapsto x$. Similarly, one can define the tangent bundle $T_p M$.
\end{definition}
Tangent and cotangent spaces of $T^*M$ split into horizontal and vertical subspaces, where the movements along the base and through the fiber of the bundle geometry take place, respectively.

\begin{definition}[Vertical and horizontal distributions]
   The tangent spaces $T_{(x,k)}$ of $T^*M$ split into vertical and horizontal subspaces as
\begin{align}
	T_{(x,k)}T^*M = \mathrm{V}_{(x,k)} \oplus \mathrm{H}_{(x,k)} = \mathrm{span}\left\{\bar\partial^\mu\right\} \oplus \mathrm{span}\left\{\delta_\mu\right\}\,, \qquad \textrm{with} \qquad \delta_\mu\, = \,\partial_\mu + N_{\nu\mu}(x,k) \bar{\partial}^\nu\,,
\end{align}
where $\partial_\mu= \partial/ \partial x^\mu$ and $\bar\partial^\mu=\partial/\partial k_\mu$. 
The vertical space can be identified physically with the momentum space and the horizontal space with the position space. The nonlinear connection (also called horizontal distribution), with coefficients $N_{\mu\nu}$, is supplementary to the vertical distribution. Analogously,  the cotangent spaces $T^*_{(x,k)}$ of $T^*M$ split into
\begin{align}
	T^*_{(x,k)}T^*M = \mathrm{V}^*_{(x,k)} \oplus \mathrm{H}^*_{(x,k)}\, =\, \mathrm{span}\left\{\delta k_\mu\right\} \oplus \mathrm{span}\left\{dx^\mu\right\}\,, \qquad \textrm{with}\qquad \delta k_\mu = d k_\mu - N_{\nu\mu}(x,k)\,dx^\nu\,.
\end{align}
\end{definition}

$T^*M$ is a metric manifold since we can define a metric on the cotangent bundle as
\begin{equation}
	\mathcal{G}= g_{\mu\nu}(x,k) dx^\mu dx^\nu+g^{\mu\nu}(x,k) \delta k_\mu \delta k_\nu\,.
\label{eq:line_element_ps} \end{equation}
Furthermore, a connection allows us to study how vector field changes from point to point in the manifold. We are interested in a connection whose covariant derivative maps vertical vectors to vertical ones (and analogously for the horizontal vectors). This is called affine connection. From it, one defines the covariant derivatives.

\begin{theorem}[Coefficients of affine connections] The affine connections on $T^*M$ can be uniquely represented in the adapted basis $(\delta_\mu, \bar\partial^\mu)$ in the following form
\begin{eqnarray}
    \begin{array}{ll}
       \hspace{4pt}\nabla_{\delta_\mu} \delta_\nu = H^\sigma{}_{\mu\nu}(x,k) \delta_\sigma\,,  &  \qquad \nabla_{\delta_\mu} \bar\partial^\nu\, =\, -H^\nu{}_{\mu\sigma}(x,k) \bar\partial^\sigma\,,\\
        \nabla_{\bar\partial^\mu} \bar\partial^\nu\, =\,- C_\sigma{}^{\mu\nu}(x,k)\bar\partial^\sigma\,, & \hspace{2pt} \qquad \nabla_{\bar\partial^\mu} \delta_\mu = C_\mu{}^{\sigma\nu}(x,k)\delta_\sigma\,,
    \end{array}
\end{eqnarray}
being $C_\mu{}^{\sigma\nu}(x,k)$ the affine connection coefficients in momentum space and $H^\sigma{}_{\mu\nu}(x,k)$ the ones of spacetime. 
\end{theorem}
\begin{proof}
    Proof can be found in~\cite{miron2001geometry}. 
\end{proof}
\begin{definition}[Covariant derivatives]
The horizontal and vertical covariant derivatives are given by
\begin{equation}
\begin{split}
T^{\mu_1\dots\mu_n}_{\nu_1\dots\nu_n;\rho}&= \delta_\rho T^{\mu_1\dots\mu_n}_{\nu_1\dots\nu_n}+T^{\lambda\mu_2\dots\mu_n}_{\nu_1\dots\nu_n}{H^{\mu_1}}_{\lambda \rho}+\dots +T^{\mu_1\dots\lambda}_{\nu_1\dots\nu_n}{H^{\mu_n}}_{\lambda \rho}-T^{\mu_1\dots\mu_n}_{\lambda\nu_2\dots\nu_n}{H^{\lambda }}_{\nu_1 \rho}-\dots T^{\mu_1\dots\mu_n}_{\nu_1\dots\lambda}{H^{\lambda }}_{\nu_n \rho} \,,\\
T^{\mu_1\dots\mu_n;\rho}_{\nu_1\dots\nu_n}&= \bar\partial^\rho T^{\mu_1\dots\mu_n}_{\nu_1\dots\nu_n}+T^{\lambda\mu_2\dots\mu_n}_{\nu_1\dots\nu_n}{C_{\lambda}}^{\mu_1 \rho}+\dots +T^{\mu_1\dots\lambda}_{\nu_1\dots\nu_n}{C_{\lambda}}^{\mu_n \rho}-T^{\mu_1\dots\mu_n}_{\lambda\nu_2\dots\nu_n}{C_{\nu_1}}^{\lambda \rho }-\dots T^{\mu_1\dots\mu_n}_{\nu_1\dots\lambda}{C_{\nu_n}}^{\lambda \rho }\,,
\end{split}
\end{equation}
respectively.
\end{definition}
In principle and up to this point, one could define infinite  connections having in general a non zero covariant derivative of the metric and being not symmetric (under the interchange of the two lower or the two upper indices, depending on whether we consider the horizontal or vertical connections, respectively). But we are interested in symmetric connections with zero covariant derivative of the metric, for (physical) simplicity. For these ones, an important result is worth highlighting. 
\begin{theorem}[Uniqueness of the affine connections]
There exists unique affine connections which are metrical, i.e.
\begin{equation}
    g_{\mu\nu;\rho}= {g_{\mu\nu}}^{;\rho}=0\,,
\end{equation}
with vanishing torsions, and with affine connection coefficients given by
\begin{align}
	{C_\rho}^{\mu\nu}(x,k)=-\frac{1}{2}g_{\rho\sigma}\left(\bar{\partial }^\mu g^{\sigma\nu}(x,k)+ \bar{\partial }^\nu g^{\sigma\mu}(x,k)-\bar{\partial }^\sigma g^{\mu \nu}(x,k)\right)\,,
	\label{eq:affine_connection_p}\\
	{H^\rho}_{\mu\nu}(x,k)=\frac{1}{2}g^{\rho\sigma}(x,k)\left(\delta_\mu g_{\sigma\nu}(x,k) +\delta_\mu g_{\sigma\mu}(x,k) -\delta_\sigma g_{\mu\nu}(x,k) \right)
	\label{eq:affine_connection_st}\,.
\end{align}
Then, the following symmetries holds
\begin{equation}
N_{\mu\nu}=N_{\nu\mu}  \,,\qquad {H^\rho}_{\mu\nu}={H^\rho}_{\nu\mu}\,,\qquad  {C_\rho}^{\mu\nu}={C_\rho}^{\nu\mu}\,.
\end{equation}
\end{theorem}
\begin{proof}
    Proof can be found in~\cite{miron2001geometry}. 
\end{proof}

In  what follows, we will use metrical affine connections. We can now focus on two set of curves: 
\begin{itemize}
    \item Vertical autoparallels: the ones with parameter $\tau$ given by $\gamma(\tau) = (x_0,k(\tau))$ which are obtained when moving along the fiber while keeping the position coordinates fixed. They will define the distance in momentum space from which the Hamiltonian is obtained.  They satisfy $\nabla_{\dot\gamma}\dot \gamma = 0$, and are solutions of the equations
    \begin{equation}
	\ddot k_\mu  - C_{\mu}{}^{\nu\sigma}(x,k) \dot{k}_\nu \dot{k}_\sigma=0\,.
	\label{eq:vertical}
    \end{equation} 
    The dot in the above expression represents the derivative with respect to the parameter $\tau$. 
    \item Horizontal autoparallels: the ones parametrized as $\gamma(\tau) = (x(\tau),k(\tau))$.  They define force-free particle motion along spacetime and will be fulfilled by solutions of the Hamilton equations of motion. They satisfy  $
	\delta \dot{k}_\lambda=\dot{k}_\lambda-N_{\sigma\lambda} (x,k)\dot{x}^\sigma=0\,,
 $ and are solutions of
    \begin{equation}
	    \ddot{x}^\mu+{H^\mu}_{\nu\sigma}(x,k) \dot{x}^\nu\dot{x}^\sigma=0\,.
	    \label{eq:horizontal_geodesics_curve_definition}
    \end{equation} 
\end{itemize}

Notice that a Hamiltonian $\mathcal{H}(x,k)$ describes the evolution of horizontal curves, changing the phase-space coordinates as
\begin{equation}
\frac{dx^\mu}{d\tau}\,=\,\mathcal{N}\lbrace{\mathcal{H} (x,k),x^\mu\rbrace}\,=\,\mathcal{N} \bar \partial^\mu \mathcal{H} (x,k)\,,\qquad\frac{dk_\mu}{d\tau}\,=\,\mathcal{N}\lbrace{\mathcal{H} (x,k),k_\mu\rbrace}\,=\,-\mathcal{N} \partial_\mu \mathcal{H} (x,k)\,,
\label{eq:H-eqs}
\end{equation}
where $\mathcal{N}$ is a Lagrange multiplier (it will take the value $1/(2m)$ or $1/2$ for massive and massless particles, respectively), and $\{,\}$ is the Poisson bracket\footnote{ We are going to choose the convention of signs for the Poisson brackets  $\lbrace{k_\nu,x^\mu\rbrace}=\delta^\mu_\nu$ to be compatible with the cotangent bundle structure defined in the previous section~\cite{miron2001geometry}.}. We now enumerate a couple of theorems that we will use in the following, stating the relationship between nonlinear and affine connections, and relating the Hamiltonian with a metric in a generalized Hamilton space.

\begin{theorem}[Relationship between Hamiltonian and metric]
\label{t8}
The geodesic distance in momentum space:
\begin{align}\label{eq:vertgeomdist}
	D(x,k) \,=\, \int_0^{\tau_1} d\tau \sqrt{g^{\mu\nu}(x,k(\tau)) \dot k_\mu(\tau) \dot  k_{\nu}(\tau)}\,,
\end{align}
through the vertical curves~\eqref{eq:vertical}, is conserved along horizontal paths if the affine connection is metrical and   
\begin{equation}
{H^\rho}_{\mu\nu}(x,k)=\bar{\partial}^\rho N_{\mu\nu}(x,k)
\label{eq:affine_connection_n}
\end{equation} 
is satisfied. Then, the squared geodesic distance in momentum space can be identified with a Hamiltonian which indeed is autoparallel, i.e. it satisfies
\begin{equation}
	\delta_\mu \mathcal{H}(x,k)=0\,,
	\label{eq:casimir_delta}
\end{equation}   
\end{theorem}
\begin{proof}
     Proof can be found in~\cite{Relancio:2020rys}. 
\end{proof}

From a physical point of view, in order to recover the Hamiltonian of SR and GR (when the metric does not depend on the momenta), we will consider the Hamiltonian to be the squared distance in momentum space\footnote{But notice that, in fact,  any function of such a (squared) distance can also be regarded as the Hamiltonian.}.  Furthermore, it will be the solution of a differential equation.  
\begin{proposition} [Differential equation for distance in momentum space]
  The squared distance in momentum space, that can be identified with the Hamiltonian $\mathcal{H}(x,k)$, satisfies the following differential equation~\cite{Relancio:2020zok} 
\begin{equation}
	\mathcal{H}(x,k)=\frac{1}{4}\bar{\partial}^\mu \mathcal{H}(x,k) g_{\mu\nu} (x,k)\bar{\partial}^\nu \mathcal{H}(x,k) \,.
	\label{eq:casimir_metric}
\end{equation} 
\end{proposition}


\subsection{Deformed relativistic kinematics and curved momentum spaces}
\label{sec:drk}
A deformed relativistic kinematics (DRK) has three basic building blocks: a deformed composition law of momenta, deformed Lorentz transformations, and a deformed dispersion relation.  DRKs can be understood in terms of non-trivial curved momentum space geometry, which is defined by a maximally symmetric metric in the momentum space of the form: $\zeta = \zeta^{\mu\nu}(k) dk_\mu dk_\nu$. The maximal symmetry requirement implies the existence of ten isometries, which consist of four translations $\mathcal{T}$, and six Lorentz transformations $\mathcal{J}$ (boosts and rotations). All the ingredients can be understood from such a geometrical setup~\cite{Carmona:2019fwf, Arzano:2022har}. Firstly, translations induce a deformed law of composition of momenta \begin{align}\label{eq:DRKflat}
	(p\oplus q)_\mu = \mathcal{T}(p,q)_\mu\,.
\end{align}
Secondly, the generators of the Lorentz transformations (rotations and boosts) are the other six generators of the isometries of the momentum metric:
\begin{align}\label{eq:DRKflat1}
 k^\prime_\mu = \mathcal{J}(k,\Omega)_\mu\,.
\end{align}
Finally, the deformed dispersion relation, or Casimir function, is obtained as the square of the minimal geometric distance from the origin of momentum space to a given momentum $k$~\cite{AmelinoCamelia:2011bm}.

Let us now implement the DRKs on each point of a curved spacetime. If the Lorentzian metric of the curved spacetime has local coordinate components $a^{\mu\nu}(x)$, the geometrically deformed one depends also on momentum, and it is determined by~\cite{Relancio:2020zok}, 
\begin{equation}
	g^{\mu \nu}(x,k)=e^{\mu}{}_\alpha (x) \zeta^{\alpha \beta}(\bar{k})e^{\nu}{}_\beta (x)\,,
	\label{eq:definition_metric_cotangent}
\end{equation}
where we introduce the components of a tetrad of the metric $ a^{\mu\nu}(x)= e^\mu{}_\alpha(x) \eta^{\alpha\beta} e^\nu{}_\beta(x)$, and define  $\bar{k}_\alpha=e^\mu{}_\alpha(x) k_\mu$. In this way, the momentum space is also maximally symmetric if the momentum metric $\zeta^{\alpha \beta}(k)$ is as well~\cite{Relancio:2020rys}.

In~\cite{Relancio:2020rys,Pfeifer:2021tas} it was shown that not every choice of momentum coordinates is compatible with the geometrical construction in the cotangent bundle. In fact, only those metrics that are linear Lorentz invariant can be lifted to curved spacetimes. For these cases, the cotangent bundle metric reduces to
\begin{align}
	g_{\mu\nu}(x,k) = a_{\mu\nu}(x)  f_1\left(\frac{\bar k^2}{\Lambda^2}\right) + \frac{1}{\Lambda^2}  k_\mu k_\nu f_2\left(\frac{\bar k^2}{\Lambda^2}\right)\,,
 \label{eq:metric_Lorentz}
\end{align}
where $\bar k^2=k_\mu k_\nu a^{\mu\nu}(x)$, ad $f_i$ are generic functions of $\bar k^2/\Lambda^2$. Notice that \eqref{eq:metric_Lorentz}  will not explicitly depend on the tetrad of the spacetime but on the spacetime metric. 

\section{Developments in Hamilton and generalized Hamilton spaces}
\label{sec3}

In this section we discuss how all the geometrical ingredients of a generalized Hamiltonian space can be obtained once a generic metric is known. This goes beyond the understanding of~\cite{miron2001geometry},  where it is not clear how to connect the nonlinear connection, affine connections, metric, and Hamiltonian in a compatible way. Moreover, we discuss new families of possible autoparallel Hamiltonians, going beyond the results of the literature~\cite{Barcaroli:2015xda}. We also establish under which conditions the generalized Hamilton spaces are Hamilton ones.

\subsection{Revisiting the definition of the nonlinear connection in generalized Hamilton spaces}
\label{subsec:construction}
In Hamilton spaces, there exists a Hamiltonian function $\mathcal{H}$ from which it is possible to construct an associated Hamilton metric as~\cite{miron2001geometry}
\begin{equation}
g_H^{\mu\nu}(x,k)= \frac{1}{2}\bar \partial^\mu \bar \partial^\nu  \mathcal{H}(x,k)  \,,
    \label{eq:H_metric}
\end{equation}
so the nonlinear connection can be computed as~\cite{miron2001geometry}
\begin{equation}
N_{\mu\nu}= -\frac{1}{4}\left(\lbrace{g^H_{\mu\nu},\mathcal{H}\rbrace} +g^H_{\mu\rho} \bar \partial^\rho  \partial_\nu   \mathcal{H}  +g^H_{\nu\rho}\bar \partial^\rho  \partial_\mu  \mathcal{H}\right)\,.
    \label{eq:N_connection}
\end{equation}
For a given Hamiltonian, the horizontal and vertical affine connections can be easily obtained from the metric and the nonlinear connection. So one is able to construct all the main geometrical ingredients of these spaces starting only from a Hamiltonian. 

Let us see what happens for generalized Hamilton spaces. The main difficulty pointed out in~\cite{miron2001geometry} is that one cannot determine neither the nonlinear connection nor the horizontal affine connection directly from the metric in these geometries, contrary to what happens in Hamilton spaces. What we propose to solve this problem, once a metric is given, is:
\begin{enumerate}
    \item Defining the affine coefficients ${C_\rho}^{\mu\nu}(x,k)$ through \eqref{eq:affine_connection_p}, since they are the only ones which can be directly determined from the metric (notice that ${H^\rho}_{\mu\nu}(x,k)$ involves $\delta$ and, thus, the nonlinear connection, which is up to now unknown).
    \item Obtaining the squared geodesic distance in momentum space from the geodesic equations \eqref{eq:vertical}, which will be identified with the Hamiltonian as discussed in Th.~\ref{t8}.
    \item Computing the nonlinear coefficients from \eqref{eq:N_connection}. Notice  that \eqref{eq:N_connection} is valid either for Hamilton and generalized Hamilton spaces since the equation relates a Hamiltonian function with the connection, no matter which kind of space is being considering. 
    \item Determining the affine coefficients ${H^\rho}_{\mu\nu}(x,k)$ using Eq.~\eqref{eq:affine_connection_n}. They will be compatible with~\eqref{eq:affine_connection_st}.
 
\end{enumerate}

This establishes a new level of understanding on generalized Hamilton spaces, since we are able to describe all their geometrical ingredients for a given metric. This step program allows working with generalized Hamilton spaces in a systematic and simple way. Since their geometrical structure is richer than those of Hamilton spaces, generalized Hamilton spaces could open a new field of study we hope to explore further in future works.

\subsection{Autoparallel Hamiltonians}
In this subsection, we discuss some properties of autoparallel Hamiltonians. 
In~\cite{Barcaroli:2015xda} it is proved that if a  Hamiltonian is an  homogeneous function in momenta, then it is always autoparallel. Here, we go one step further and  state the following result:
\begin{theorem}
If a given Hamiltonian $\mathcal{H}$ is autoparallel, any function of it $\mathcal{\tilde H}=\mathcal{\tilde H}(\mathcal{H})$, will be also autoparallel. 
\label{theoremIII2}
\end{theorem}
\begin{proof}
   It lies in the fact that 
\begin{equation}
	\delta_\mu \mathcal{\tilde H}=\partial_\mu \mathcal{\tilde H}+N_{\nu\mu}\bar \partial^\nu \mathcal{\tilde H}=\frac{\partial \mathcal{\tilde H}}{\partial \mathcal{H}}\delta_\mu \mathcal{ H}=0\,.
 \label{eq:proof_hfunc}
\end{equation} 
\end{proof}

Note that in a Hamilton space, all the geometrical ingredients can be derived from the starting Hamiltonian. However, the squared distance in momentum space will in general not coincide with the starting Hamiltonian. But, we can assure that such a distance will be a function of it, since both quantities are conserved along horizontal curves.  It is important to drawn attention to the fact that, if one wants a deformed relativistic kinematics with an autoparallel Hamiltonian, showing a smooth limit to SR so when the deformation parameter $\Lambda$ goes to infinity one recovers the Hamiltonian of SR (i.e. $\mathcal{H}=k^2$), the only possible deformed Hamiltonian is a function of the original one. 

In the rest of the subsection, we discuss some properties of the geometrical objects 
for any kind of autoparallel Hamiltonians. In particular, we focus on Hamiltonians which are homogeneous functions of the momenta due to the particularly simple relations that can be obtained between the connections and the \textit{d-curvature} and horizontal curvature tensors. But before addressing such relations, it is necessary to introduce some concepts. On the one hand, the \textit{d-curvature} tensor is defined as~\cite{miron2001geometry}
\begin{equation}
R_{\mu\nu\rho}\,=\, \delta_\rho N_{\nu\mu}-\delta_\nu N_{\rho\mu}\,.
\label{eq:dtensor}
\end{equation} 
On the other hand,  the horizontal curvature tensor can be written\footnote{Note that the definition given in~\cite{Relancio:2020rys}, mandatory to find some consistent Einstein's equations compatible with a momentum independent energy-momentum tensor, differs from the one of~\cite{miron2001geometry}, which is instead proposed for its mathematical simplicity.  } as~\cite{Relancio:2020rys}
 \begin{equation}
{R^\mu}_{\nu \rho\sigma}\,=\, \delta_\sigma {H^\mu}_{\nu \rho} - \delta_\rho {H^\mu}_{\nu \sigma} +{H^\lambda}_{\nu \rho}{H^\mu}_{\lambda \sigma}-{H^\lambda}_{\nu \sigma}{H^\mu}_{\lambda \rho}\,.
\label{riemann_st}
\end{equation}
Now we have all the milestones to focus on various specific properties of some geometrical objects for autoparallel Hamiltonians.

\begin{theorem}
\label{th:hhtilde}
The nonlinear coefficients, the horizontal affine connection, and the horizontal curvature tensor of a pair of generalized Hamilton spaces with Hamiltonians $\mathcal{H}$ and $\mathcal{\tilde H}=\mathcal{\tilde H}(\mathcal{H})$, respectively, are the same. 
\end{theorem}
\begin{proof}
Following the discussion of Sec.~\ref{subsec:construction}, both geometries have the same nonlinear coefficients, since they are compatible with the same Hamiltonian. Moreover, if the Hamiltonian is autoparallel, we know that Eq.~\eqref{eq:affine_connection_n} holds. This means that the horizontal affine connection must be the same for both kind of spaces and, therefore, the horizontal curvature tensor must be also the same.
\end{proof}

\begin{corollary}
The nonlinear coefficients, the horizontal affine connection, and the horizontal curvature tensor of a pair of generalized and regular Hamilton spaces, both of them compatible with the same autoparallel Hamiltonian $\mathcal{H}$, are the same. 
\end{corollary}
\begin{proof}
This is a particular case of the previous theorem.
\end{proof}

\begin{theorem}
If a Hamiltonian $\mathcal{H}$ is autoparallel, the horizontal affine connection satisfies
\begin{equation}
\bar\partial^\sigma {H^\rho}_{\mu \nu}\bar\partial^\nu \mathcal{H}=0\,.
\label{eq:H_derivative_H}
\end{equation} 
\end{theorem}
\begin{proof}
By deriving twice with respect to momenta the condition $\delta_\mu \mathcal{H}=0$, and using Eq.~\eqref{eq:affine_connection_n}, one finds
\begin{equation}
\partial_\mu g_H^{\rho \sigma}+N_{\mu\nu}\bar\partial^\sigma g_H^{\rho \nu}+\frac{1}{2}\bar\partial^\sigma {H^\rho}_{\mu \nu}\bar\partial^\nu \mathcal{H}+ {H^\rho}_{\mu \nu}g_H^{\nu \sigma} + {H^\sigma}_{\mu \nu} g_H^{\nu \rho}=0\,.
\end{equation} 
Since $N_{\mu\nu}$ and ${H^\rho}_{\mu \nu}$ are symmetric under the exchange of $\mu$ and $\nu$, and due to the fact that 
\begin{equation}
    \bar\partial^\sigma g_H^{\rho \nu}= \bar\partial^\nu g_H^{\rho \sigma}\,,
\end{equation}
which can be easily seen from the definition~\eqref{eq:H_metric} of $g_H$, the previous equation can be written as 
\begin{equation}
{g^{\rho \sigma}_{H}}_{;\mu} +\frac{1}{2}\bar\partial^\sigma {H^\rho}_{\mu \nu}\bar\partial^\nu \mathcal{H}=0\,.
\end{equation} 
Since we are considering a metrical compatible connection, the covariant derivative of the metric vanishes, and thus the last term of the above equation must be zero.
\end{proof}

\begin{theorem}
\label{th_III.4}
If an autoparallel Hamiltonian $\mathcal{H}$ is a $r$-homogeneous function on momenta, or an  arbitrary function of a homogeneous function on momenta, then 
\begin{equation}
	N_{\mu \nu}=k_\rho {H^\rho}_{\mu \nu}
 \label{eq:nkh}
\end{equation} 
holds.
\end{theorem}
\begin{proof}
Let $\mathcal{H}$ be a $r$-homogeneous function on momenta. Thus, by definition~\cite{Barcaroli:2015xda,miron2001geometry}
\begin{equation}
	r \mathcal{H}=k_\rho \bar\partial^\rho \mathcal{H}\,
 \label{eq:homogeneousH}
\end{equation}
holds. Therefore, if our Hamiltonian is a function $\mathcal{\tilde H}(\mathcal{H})$ of a homogeneous function of degree $r$ on momenta, we have
\begin{equation}
k_\rho \bar\partial^\rho \mathcal{\tilde H}=k_\rho \bar\partial^\rho \mathcal{H} \frac{\partial \mathcal{\tilde H}}{\partial \mathcal{H}}\,.
\end{equation} 
Since in this case the Hamiltonian is autoparallel, it satisfies
\begin{equation}
0=\delta_\mu \mathcal{\tilde H}=\delta_\mu \left(k_\rho \bar\partial^\rho \mathcal{H} \frac{\partial \mathcal{\tilde H}}{\partial \mathcal{H}}\right)\,.
\end{equation} 
We can rewrite this last equation as
\begin{equation}
0= \delta_\mu k_\rho \bar\partial^\rho \mathcal{H} \frac{\partial \mathcal{\tilde H}}{\partial \mathcal{H}} + k_\rho \delta_\mu \bar\partial^\rho \mathcal{H} \frac{\partial \mathcal{\tilde H}}{\partial \mathcal{H}}  + k_\rho \bar\partial^\rho \mathcal{H} \delta_\mu\frac{\partial \mathcal{\tilde H}}{\partial \mathcal{H}} = \left(N_{\mu\sigma} \bar\partial^\sigma \mathcal{H} -k_\rho \bar\partial^\rho N_{\mu\sigma}\bar\partial^\sigma \mathcal{H}\right)\frac{\partial \mathcal{\tilde H}}{\partial \mathcal{H}}\, ,
\end{equation} 
where we have used that~\cite{miron2001geometry}
\begin{equation}
    [\delta_\mu,\bar\partial^\rho]=-\bar\partial^\rho N_{\mu\lambda}\bar\partial^\lambda\,.
\end{equation}
and the fact that
\begin{equation}
   \delta_\mu\frac{\partial \mathcal{\tilde H}}{\partial \mathcal{H}}= \frac{\partial^2 \mathcal{\tilde H}}{\partial \mathcal{H}^2}  \delta_\mu \mathcal{H}=0\,.
\end{equation}
Thus, we obtain the relation:
\begin{equation}
    \bar\partial^\sigma \mathcal{H} \left(N_{\mu\sigma}-k_\rho {H^\rho}_{\mu\sigma}\right)=0\, .
\end{equation}
By deriving the previous expression with respect to momentum, one finds
\begin{equation}
0= 2 g^{\sigma\lambda}_H \left(N_{\mu\sigma}-k_\rho {H^\rho}_{\mu\sigma}\right)+\bar\partial^\sigma \mathcal{H}\left(\bar\partial^\lambda{N}_{\mu\sigma}-{H^\lambda}_{\mu\sigma}-k_\rho \bar\partial^\lambda{H^\rho}_{\mu\sigma}\right)=0 \,.
\end{equation}
By using Eqs.~\eqref{eq:affine_connection_n},\eqref{eq:H_derivative_H} one arrives to Eq.~\eqref{eq:nkh}.

\end{proof}

\begin{theorem}
If an autoparallel Hamiltonian $\mathcal{H}$ is a homogeneous function on momenta, or an  arbitrary function of a homogeneous function on momenta, the nonlinear connection and the horizontal affine connection are homogeneous on momentum of degree 1 and 0, respectively.
\end{theorem}
\begin{proof}
From Eqs.~\eqref{eq:nkh} and \eqref{eq:affine_connection_n}, we know that 
\begin{equation}
    N_{\mu\nu}=k_\rho {H^\rho}_{\mu\nu}=k_\rho \bar \partial^\rho N_{\mu\nu}\,.
\end{equation}
Taking into account the definition \eqref{eq:homogeneousH}, one easily finds out that the nonlinear connection must be homogeneous on momentum of degree 1. Furthermore, by deriving the previous expression with respect to momentum, we see that
\begin{eqnarray}
       \bar\partial^\lambda   N_{\mu\nu}&=& \bar\partial^\lambda  (k_\rho {H^\rho}_{\mu\nu})\\
       {H^\lambda}_{\mu\nu}&=& {H^\lambda}_{\mu\nu}+k_\rho \bar\partial^\lambda {H^\rho}_{\mu\nu}\\
      k_\rho \bar\partial^\lambda {H^\rho}_{\mu\nu}&=&k_\rho \bar\partial^\rho {H^\lambda}_{\mu\nu}=0 \,,
\end{eqnarray}
where in the second and third steps we have used Eq.~\eqref{eq:affine_connection_n}. Finally, by using again the definition \eqref{eq:homogeneousH}, we conclude that the horizontal affine connection must be homogeneous on momentum of degree 0.
\end{proof}

\begin{theorem}\label{th3.7}
If an autoparallel Hamiltonian $\mathcal{H}$ is a homogeneous function of degree 2 on momenta, the connections of a Hamilton space defined by it are of the Cartan type~\cite{miron2001geometry}, i.e.
\begin{equation}
	N_{\mu \nu}=k_\rho {H^\rho}_{\mu \nu}\,,\qquad k_\lambda {C_\mu}^{\lambda\nu}=0\,.
 \label{eq:nlinearcartan}
\end{equation} 
\end{theorem}
\begin{proof}
If $\mathcal{H}$ is a homogeneous function on momenta, we know from the previous theorem that the first equation in \eqref{eq:nlinearcartan} is satisfied. 
In order to prove the second equality, by deriving Eq.~\eqref{eq:homogeneousH} with respect to momentum one finds
\begin{equation}
	\bar\partial^\mu \mathcal{H}+2 k_\rho g^{\rho\mu}_H=r \bar\partial^\mu \mathcal{H}\,,
\end{equation} 
so 
\begin{equation}
	\frac{r-1}{2} \bar\partial^\mu \mathcal{H}=k_\rho g^{\rho\mu}_H\,.
\end{equation} 
By deriving the previous equation once again with respect to momentum and using that, for Hamilton spaces, the vertical affine connection is given by\footnote{This can be easily obtained from the symmetries of the momentum derivatives of the metric of the Hamilton space~\eqref{eq:H_metric}.} ~\cite{miron2001geometry}
\begin{equation}
   {C_\rho}^{\mu\nu}=-\frac{1}{2} \bar\partial_\rho g^{\mu\nu}_H\,,
   \label{eq:c_hamilton}
\end{equation}
one obtains
\begin{equation}
	(r-2) g^{\mu\nu}_H+2 k_\rho C^{\rho\mu\nu}=0\,.
\end{equation} 
Therefore, Eq.~\eqref{eq:nlinearcartan} is satisfied if $r=2$.
\end{proof}

Let us study now some properties of the horizontal curvature tensor for autoparallel Hamiltonians.

\begin{theorem}
If a Hamiltonian $\mathcal{H}$ is autoparallel, 
\begin{equation}
	\bar \partial^\sigma R_{\mu \nu \rho}= {R^\sigma}_{\mu \nu \rho}
 \label{eq:rdr}
\end{equation} 
holds.
\end{theorem}
\begin{proof}
By expanding Eq.~\eqref{eq:dtensor} and using Eq.~\eqref{eq:affine_connection_n}, one finds 
  \begin{equation}
       R_{\mu \nu \rho}= \partial_\rho N_{\nu\mu}-\partial_\nu N_{\rho\mu}+N_{\lambda\rho} {H^\lambda}_{\nu\mu}-N_{\lambda\nu} {H^\lambda}_{\rho\mu}\,.
  \end{equation}
  By deriving this expression with respect to momentum one finds
  \begin{eqnarray}
      \bar\partial^\sigma R_{\mu \nu \rho} &=& \partial_\rho {H^\sigma}_{\nu\mu}-\partial_\nu {H^\sigma}_{\rho\mu}+{H^\sigma}_{\lambda\rho} {H^\lambda}_{\nu\mu}+N_{\lambda\rho}  \bar\partial^\sigma {H^\lambda}_{\nu\mu}-{H^\sigma}_{\lambda\nu} {H^\lambda}_{\rho\mu}-N_{\lambda\nu}  \bar\partial^\sigma {H^\lambda}_{\rho\mu}\\
      &=& \delta_\rho {H^\sigma}_{\nu\mu}-\delta_\nu {H^\sigma}_{\rho\mu}\delta+{H^\sigma}_{\lambda\rho} {H^\lambda}_{\nu\mu}-{H^\sigma}_{\lambda\nu} {H^\lambda}_{\rho\mu}= {R^\sigma}_{\mu \nu \rho}\,.
  \end{eqnarray}
\end{proof}

\begin{theorem}
If an autoparallel Hamiltonian $\mathcal{H}$ is a homogeneous function on momenta, 
\begin{equation}
	R_{\mu \nu \sigma}=k_\rho {R^\rho}_{\mu \nu \sigma}
 \label{eq:rkr}
\end{equation} 
holds.
\end{theorem}
\begin{proof}
From  Eq.~\eqref{eq:nkh} it is easy to see that
\begin{equation}
\delta_\rho N_{\mu\nu}=\delta_\rho\left(k_\sigma {H^\sigma}_{\mu\nu}\right)= N_{\rho\sigma} {H^\sigma}_{\mu\nu}+k_\sigma \delta_\rho {H^\sigma}_{\mu\nu}=k_\sigma \left( 
{H^\sigma}_{\rho\lambda}{H^\lambda}_{\mu\nu}+ \delta_\rho {H^\sigma}_{\mu\nu}\right) \,,\label{eqdeltaN1}
\end{equation} 
where we have used Eqs.~\eqref{eq:affine_connection_n} and~\eqref{eq:nkh}. Similarly, we could write
\begin{eqnarray}
 \delta_\nu N_{\rho\mu}=k_\sigma \left(
{H^\sigma}_{\beta\nu}{H^\beta}_{\rho\mu}+ \delta_\nu {H^\sigma}_{\rho\mu}\right)\,. \label{eqdeltaN2}
\end{eqnarray}
Now, plugging Eqs. \eqref{eqdeltaN1} and \eqref{eqdeltaN2} into the definition of the \textit{d-curvature} tensor of Eq.~\eqref{eq:dtensor}, one finds Eq.~\eqref{eq:rkr}. 
\end{proof}

\begin{theorem}
If an autoparallel Hamiltonian $\mathcal{H}$ is a homogeneous function on momenta, or an  arbitrary function of a homogeneous function on momenta, the \textit{d-curvature} and the horizontal curvature tensors are homogeneous on momentum of degree 1 and 0, respectively.
\end{theorem}
\begin{proof}
By deriving Eq.~\eqref{eq:rkr} with respect to momentum one finds
\begin{equation}
	\bar \partial^\lambda R_{\mu \nu \sigma}= {R^\lambda}_{\mu \nu \sigma}+k_\rho \bar \partial^\lambda  {R^\rho}_{\mu \nu \sigma}\,.
\end{equation} 
By using Eq.~\eqref{eq:rdr} one obtains
\begin{equation}
	k_\rho \bar \partial^\lambda  {R^\rho}_{\mu \nu \sigma}=	k_\rho \bar \partial^\rho  {R^\lambda}_{\mu \nu \sigma}=0\,.
\end{equation} 
From here we see that the horizontal curvature tensor is homogeneous on momentum of degree 0. Then, from Eqs.~\eqref{eq:rdr}-\eqref{eq:rkr} and the definition \eqref{eq:homogeneousH} we see that the \textit{d-curvature}  tensor is homogeneous on momentum of degree 1, since $k_\rho \bar \partial^\rho R_{\mu\nu\sigma}=k_\rho {R^\rho}_{\mu\nu\sigma}= R_{\mu\nu\sigma}$.
\end{proof}

\subsection{Connection between regular and generalized Hamilton spaces}
As discussed in~\cite{miron2001geometry}, every Hamilton space can be regarded as a particular case of a generalized Hamilton one (but not the other way around). The main difference between both spaces is the fact that, due to symmetries, the vertical affine connection in Hamilton spaces takes a simplified expression (Eq.~\eqref{eq:c_hamilton}) with respect to that for generalized Hamilton spaces (Eq.~\eqref{eq:affine_connection_p}).


The next question that naturally arises now is to think about what are the components of a metric tensor, corresponding to a maximally symmetric space, and which satisfies Eq.~\eqref{eq:c_hamilton}. Knowing them, we would have obtained a maximally symmetric space corresponding to a generalized Hamilton space. Then, the components of the metric can be obtained from a Hamiltonian by using \eqref{eq:H_metric}. This will restrict the form of the components of the metric, leading to a unique possible solution (note that the components of the vertical affine connection changes in a non-tensorial way under changes of momentum coordinates). 



Following the discussion of the previous subsection, let us consider a generic Hamiltonian of the form $\mathcal{H}=\mathcal{H}(k^2)$. We choose a Hamiltonian depending on the momentum squared in order to obtain a metric of the form~\eqref{eq:metric_Lorentz}.  From it we can express the metric in terms of this ansatz for the Hamiltonian by using \eqref{eq:H_metric}. Once calculated the metric, one can rewrite the $C_{\sigma}^{\mu\nu}(k)$ coefficients \eqref{eq:c_hamilton} and the curvature tensor in momentum space~\cite{miron2001geometry}  
\begin{equation}
{S_{\sigma}}^{\mu\nu\rho}= \frac{\partial {C_\sigma}^{\mu\nu}}{\partial k_\rho}-\frac{\partial {C_\sigma}^{\mu\rho}}{\partial k_\nu}+{C_\sigma}^{\lambda\nu} \,{C_\lambda}^{\mu\rho}-{C_\sigma}^{\lambda\rho}\,{C_\lambda}^{\mu\nu}\,,
\label{eq:Riemann_p}
\end{equation} 
as functions of $\mathcal{H}$ and its derivatives. On the other hand, since we want a maximally symmetric space (we will choose a de Sitter metric to make calculations simpler), we know that the following equation holds
\begin{equation}
S^{\tau \mu \nu \rho}= \frac{1}{\Lambda^2}\left( g^{\tau \nu}g^{\kappa \mu}-g^{\tau \kappa}g^{\mu \nu}\right)\,,
\label{eq:momentum_gs}
\end{equation} 
leading to a momentum curvature equal to $12/\Lambda^2$.  Again, we can write the metric in~\eqref{eq:momentum_gs} in terms of $\mathcal{H}$ by means of \eqref{eq:H_metric}. Now, the resulting expressions \eqref{eq:Riemann_p} and \eqref{eq:momentum_gs} can be equated to each other and a system of differential equation in $\mathcal{H}$ results. Imposing that  the Hamiltonian is a function of $k^2$ (so it is compatible with a metric of the family of~\eqref{eq:metric_Lorentz}) the previous system of differential equations can be solved, finding that  the Hamiltonian is
\begin{equation}
    \mathcal{H}(k^2)= \pm \Lambda \sqrt{k^2}-\frac{1}{2}\Lambda^2 \ln \left(1\pm \frac{2 \sqrt{k^2}}{\Lambda}\right)\,.
    \end{equation}
Now, using Eq.\eqref{eq:H_metric} one can find the corresponding de Sitter momentum metric.  
    

\section{Isometries of spacetime in a cotangent bundle metric}\label{sec4}
Previously, we have required a metric provided with ten isometries in the momentum space. The next step is to consider what happens to those isometries in spacetime. As we will see in this section, not every cotangent bundle geometry admits isometries in spacetime. In particular, we discuss the isometry algebra for the de Sitter case in GR and deformed GR (DGR) in spacetime.

\subsection{Spacetime isometries of the metric}

In order to understand what follows, there are two important definitions that are worth collecting. 
\begin{definition}[Canonical symplectic structure of $T^*M$]
The pair $(T^*M, dp_i \wedge dx^i)$ is a symplectic manifold where changes of variables are given by~\cite{miron2001geometry}
\begin{equation}
x^{\prime\mu}=x^{\prime\mu}(x)\,,\qquad k^\prime_{\mu}= \frac{\partial x^{\nu}}{\partial x^{\prime\mu}} k_\nu\,,
\end{equation}
implying
\begin{equation}
dx^{\prime\mu}=\frac{\partial x^{\prime\mu}}{\partial x^{\nu}} dx^\nu\,,\qquad \delta k^\prime_{\mu}= \frac{\partial x^{\nu}}{\partial x^{\prime\mu}} \delta k_\nu\,.
\end{equation}
\end{definition}

\begin{definition}[Spacetime isometry of the metric]
Any transformation that preserves the following phase-space line element:
\begin{equation}
	\mathcal{G}= g_{\mu\nu}(x,k) dx^\mu dx^\nu+g^{\mu\nu}(x,k) \delta k_\mu \delta k_\nu=g_{\mu\nu}(x^\prime,k^\prime) dx^{\prime\mu} dx^{\prime\nu}+g^{\mu\nu}(x^\prime,k^\prime) \delta k^\prime_{\mu} \delta k^\prime_{\nu}\,,
\end{equation}
is a space-time isometry of the metric.
\end{definition}
Therefore, these transformations indeed imply that 
\begin{equation}\label{isometrydef}
    \frac{\partial x^{\prime\mu}}{\partial x^\rho}g_{\mu\nu}(x^\prime,k^\prime)\frac{\partial x^{\prime\nu}}{\partial x^\sigma}=g_{\rho\sigma} (x,k)\,.
\end{equation}
Notice that the squared distance in momentum space, which can be identified with the Hamiltonian and satisfies~\eqref{eq:casimir_metric}, is indeed a Casimir of the algebra of isometries in spacetime.  This can be seen directly from the fact that the distance function does not change under isometries (since the metric neither does), so it is invariant under these transformations. This is a valid reasoning for any cotangent bundle metric depending on both spacetime and momentum, and, in particular, for GR. Let us consider now a specific example where this applies in GR.

\subsection{General Relativity}
We consider the de Sitter space (in 1+1 dimensions, for simplicity) and its algebra of isometry generators\footnote{The group of isometries of the (D+1)-dimensional de Sitter spacetime is called de Sitter group, and its connected component is isomorphic to the orthogonal group $SO(D+1,1)$ in Lorentzian signature (see \cite{Herranz_2008} and references therein). }.  This algebra is given by \cite{Herranz_2008}
\begin{equation}
\lbrace{E,P\rbrace}= \alpha P\,,\qquad \lbrace{P,N\rbrace}= E\,,\qquad \lbrace{E,N\rbrace}= P-\alpha N\,, 
\end{equation}
where $E$, $P$, and $N$ are the generators of time translations, space translations, and boost, respectively. Moreover, $\alpha$ is the cosmological constant. The Casimir of this algebra (i.e. the operator that commutes with every generator of the algebra) is given by
\begin{equation}
\mathcal{C}=E^2-P^2+2 \alpha N P\,.
\label{eq:def_casimir_algrebra}
\end{equation}
For the de Sitter coordinates in which the metric takes the form
\begin{equation}
    a_{00}=1\,,\qquad a_{01}=a_{10}=0\,,\qquad a_{11}=e^{2 \alpha x^0}\,,\qquad 
\end{equation}
the generators of isometries are \cite{Pascu:2012yu, Weinberg:1972kfs}
\begin{equation}
    E=k_0-\alpha k_1 x^1\,,\qquad P=k_1\,,\qquad N=x^1 k_0-k_1\left(\frac{\alpha (x^1)^2}{2}+\frac{1}{2\alpha}\left(e^{-2 \alpha x^0}-1\right)\right)\,.
\label{eq:generators_ds_gr}
\end{equation}
Since the metric does not depend on momenta, the squared distance in momentum space is simply
\begin{equation}
    \mathcal{H}(x,k)=\bar k^2=k_\mu a^{\mu\nu}k_\nu=k_0^2-k_1^2 e^{-2 \alpha x^0}\,.
\end{equation}
Now, it is easy to check that exactly the same expression is found when substituting the representation of the generators given in Eq.~\eqref{eq:generators_ds_gr} into~\eqref{eq:def_casimir_algrebra}, and thus $\mathcal{C}=\mathcal{H}(x,k)$, showing a particular realization of our derivation. 

This is a particular example of a curved spacetime. For different spaces one would then have different isometry algebras, but the Casimir of those algebras can be always identified with the squared distance in momentum space. This is also valid for a generic metric in cotangent bundle, depending on all phase-space variables. 

\subsection{Deformed General Relativity}
A generic (momentum and spacetime coordinates dependent) metric in DGR can be expressed as
\begin{equation}
    g_{\mu\nu}(x,k)=a_{\mu\nu}(x) f_1(x,k)+ \frac{f_2(x,k)}{\Lambda^2}k_\mu k_\nu+ \frac{f_3(x,k)}{\Lambda}(k_\mu Z_\nu+k_\nu Z_\mu)+f_4(x,k) Z_\mu Z_\nu\,,
\label{eq:generic_metric_dgr}
\end{equation}
where $Z_\alpha= n_\lambda {e^\lambda}_\alpha(x)$ is a timelike vector (which is normally used for constructing DRKs, as we will see in the next section), and the functions $f_i$ are such that $f_1(x,0)=1$ and $f_4(x,0)=0$, in order to recover the GR metric when taking the limit $\Lambda \to \infty$. Imposing that we have the same isometries as in general relativity (so that the limiting case recovers such a theory) we will see which forms of the metric~\eqref{eq:generic_metric_dgr} are allowed.  Considering the same isometries as GR and not a small number of them is physically motivated by the fact that it is desirable that several properties of the metric, such as spatial isotropy, are preserved, since it is important in theories ranging from black holes to cosmological models~\cite{Weinberg:1972kfs}. If we look for generic isometries leaving invariant~\eqref{eq:generic_metric_dgr}, the condition \eqref{isometrydef} must hold. Furthermore we can see an isometry as a change of spacetime coordinates $x \to x'(x)$ in the fiber bundle which does not depend on  momentum.  If we look at the first term of~\eqref{eq:generic_metric_dgr}, we see that the only way in which $a_{\mu\nu}(x) f_1(x,k)$ remains invariant under isometries is that $f_1(x',k')=f_1(x,k)$. Hence, $f_1$ must only depend on the momentum distance induced by $a_{\mu\nu}(x)$ because $\bar k^2$ is invariant (since it is a function of the Hamiltonian). The same reasoning could be applied for the function $f_2(x,k)$ in the second term. Moreover, the product $k_\mu k_\nu$ in the second term transforms appropriately, since~\cite{miron2001geometry} 
\begin{equation}
    k^\prime_\mu=\frac{\partial x^\nu}{\partial x^{\prime\mu}}k_\nu\,
\end{equation}
is fulfilled. The other two terms in~\eqref{eq:generic_metric_dgr} involve $Z_\mu$, but the tetrad is not generically invariant under isometries (so  neither is $Z_\mu$). Thus $f_3(x,k)=f_4(x,k)=0$. Consequently, this restricts the possible form of the metric under consideration, limiting us to those preserving linear Lorentz invariance in momentum space\footnote{
For constructing a DRK we are imposing a maximally symmetric momentum space, so Lorentz group is a part of the isometry group (see Sec.~\ref{sec:drk}). In particular, if one wants that the same isometries of GR appears in DGR, linear Lorentz invariance in momentum space must be preserved.  However, in spacetime such condition is not mandatory since for generic curved spacetimes there are no isometries corresponding to the Lorentz group (only locally can such an invariance be found).}
, i.e. to~\eqref{eq:metric_Lorentz}.  Therefore, both GR and DGR will have the same isometry algebra and, hence, the same Casimirs. The reason is that the squared distance in momentum space for the metrics~\eqref{eq:metric_Lorentz} are 
functions of the Casimir of GR~\cite{Pfeifer:2021tas}.



Now, one might wonder if there is any particular choice of the components of the metric \eqref{eq:generic_metric_dgr} so some isometries exist (see ~\cite{Relancio:2020zok} for a particular example). The answer is yes, it could be the case for particular choices (see~\cite{Gonner:1997ec,Caponio:2017lgy,Voicu:2023zem} for discussions about examples of Finsler geometries with isometries) but not in the general case. Nevertheless, we are interested here in establishing a model valid for any coordinate system, so we have to restrict ourselves to coordinate frames which are Lorentz invariant in momentum.  Of course, we do not mean to imply that there are no generic transformations that involve the momentum, $x^{\prime \mu}=x^{\prime \mu}(x,k)$ and $k^{\prime}_{\mu}=k^{\prime}_{\mu}(x,k)$, preserving the canonical Poisson bracket structure, $\lbrace{x^{\prime \mu},k^{\prime}_{\nu}\rbrace}=\delta^\mu_\nu$, and leaving the line element in phase space invariant. However, this kind of transformations are not considered in the literature before and such exploration goes beyond the scope of this paper.

Note that the aforementioned restriction on the form of the generic metric~\eqref{eq:generic_metric_dgr}  only arises when introducing a curved spacetime. If the metric does not depend on the spacetime coordinates, the generators of translations will satisfy the same algebra of SR
\begin{equation}
\lbrace{E,P\rbrace}= 0\,. 
\end{equation}
Consequently, whenever the metric~\eqref{eq:generic_metric_dgr} corresponds to a maximally symmetric  momentum space, one will be able to find some (in general) nonlinear Lorentz transformations, corresponding to isometries of the metric (they will be linear if the metric does not depend on a fixed vector).  The degeneracy present for flat spacetimes and the restriction for curved ones was also discussed in~\cite{Relancio:2020rys}.

\subsection{Choice of momentum metrics}
We have seen that there is an arbitrariness in the choice of the component representation of the metric: any metric of the form of~\eqref{eq:metric_Lorentz} is compatible with the consistent lift to curved spacetimes~\cite{Pfeifer:2021tas}, having the same isometries of the momentum independent metric. 

Now we follow a physical argument to select only one. As discussed in~\cite{miron2001geometry}, there are three different curvature tensors: one associated to spacetime,  one of momentum space, and another one of the interwined spacetime and momentum spaces. As it is well known in GR, Einstein's equations relate the energy-momentum tensor to the  space-time curvature tensor~\cite{Weinberg:1972kfs}. Following the same idea, a generalization of Einstein's equations was proposed in~\cite{miron2001geometry}, having then three different energy-momentum tensors, one associated to each curvature tensor. 

Physically, an energy-momentum tensor in spacetime can be easily understood: the content of matter and radiation bends spacetime, leading to a nontrivial geometry. In momentum space, the corresponding energy-momentum tensor is associated to a cosmological constant. However, a source for the intertwining curvature tensor is very difficult to explain from a physical point of view. Nevertheless, which is clear is that if the intertwining curvature tensor is zero, its corresponding energy-momentum tensor also vanishes. 


As shown in~\cite{miron2001geometry},  the intertwining curvature tensor is given by
\begin{equation} 
{P^{\mu\rho}}_{\lambda \nu}\,=\, \bar\partial^\rho {H^\mu}_{\lambda\nu} -{C^{\mu\rho}}_{\lambda;\nu}+{C^{\mu\sigma}}_{\lambda}{P^\rho}_{\sigma\nu}\,,
\label{eq:Riemann_sp}
\end{equation} 
being 
\begin{equation} 
{P^\rho}_{\sigma\nu}\,=\,{H^\rho}_{\sigma\nu}-\bar\partial^\rho N_{\nu\sigma}\,.
\label{eq:P_tensor}
\end{equation}  
Therefore, we can obtain the following result.

\begin{theorem}
For the family of metrics given in~\eqref{eq:metric_Lorentz}, if $f_2=0$ then  ${P^{\mu\rho}}_{\lambda \nu}=0$ is satisfied. Therefore, the metric is a conformal one  
\begin{align}
	g_{\mu\nu}(x,k) = a_{\mu\nu}(x)  f_1\left(\frac{\bar k^2}{\Lambda^2}\right) \,,
 \label{eq:conformal_metric}
\end{align}
depending the conformal factor on $\bar k^2$.
\end{theorem}
\begin{proof}
  For the kind of metrics of~\eqref{eq:metric_Lorentz}, we know that the Hamiltonian is a function of $\bar k^2$, and thus autoparallel~\cite{Barcaroli:2015xda}. Therefore, by Eq.~\eqref{eq:affine_connection_n}, ${P^\rho}_{\sigma\nu}=0$. Also, from Th.~\ref{th:hhtilde} we know that $\bar\partial^\rho {H^\mu}_{\lambda\nu}=0$, since in GR $N_{\mu\nu}=k_\lambda {\Gamma^\lambda}_{\mu\nu}$~\cite{miron2001geometry} holds, where ${\Gamma^\lambda}_{\mu\nu}$ are the Christoffel symbols (which are independent on $k$). Hence, we only need to prove that the covariant derivative of the vertical affine connection vanishes for that type of metric. 

  Plugging~\eqref{eq:conformal_metric} into the definition~\eqref{eq:affine_connection_p}, it is easy to find
  \begin{equation}
  {C_{\mu}}^{\nu\rho}=\frac{f_1^\prime}{f_1}\left(k_\tau a^{\nu\tau}\delta^{\rho}_\mu+k_\tau a^{\rho\tau}\delta^{\nu}_\mu-k_\mu a^{\nu\rho}\right)\,,
\end{equation}
where $f_1^\prime$ denotes its derivative with respect to $\bar k^2$. Since the covariant derivative of the metric vanishes (by definition), and we know from \ref{th3.7} that
\begin{equation}
    k_{\mu;\nu}=N_{\mu\nu}-{H^\lambda}_{\mu\nu}k_\lambda=0\,,
\end{equation}
the covariant derivative of the vertical affine connection is also zero. Then, ${P^{\mu\rho}}_{\lambda \nu}=0$.
\end{proof}

Indeed, since the components of the vertical affine connection do not transform tensorially under momentum changes of coordinates, the horizontal  covariant derivative of this connection must be zero only for the metric~\eqref{eq:conformal_metric}.

As a comment of the above theorem, notice that for the particular case of a maximally symmetric momentum space, $f_1$ is unequivocally determined, being
\begin{equation}
  f_1=\left(1\pm\frac{\bar k^2}{4\Lambda^2}\right)^2\,,
\end{equation}
where the plus and minus signs correspond to AdS and dS, respectively.

\section{Deformed Special and general Relativity and isometries of momentum space}\label{sec5}
In the previous section, we discussed the isometries of a cotangent bundle geometry in spacetime. Now we discuss the isometries in momentum space and its relationship with nonlocality. But first of all, we want to review some previous useful results.

\begin{definition}
Starting from a cotangent bundle metric, one can define some transformations on the vertical bundle, i.e. for a fixed point on the base manifold (spacetime), that change the position in the fiber (momentum). This transformations are defined as
\begin{equation}\label{isometrydefmomentum}
  g_{\mu\nu}(x,k^\prime)=  \frac{\partial k^{\prime}_\mu}{\partial k_\rho}g_{\rho\sigma} (x,k) \frac{\partial k^{\prime}_\nu}{\partial k_\sigma}\,.
\end{equation}
\end{definition}

These transformations will generally depend on the point of the base manifold. Indeed, as discussed in~\cite{Pfeifer:2021tas}, for the particular case of the (maximally symmetric momentum) metric 
\be
 g_{\mu\nu}(k)=a_{\mu\nu} + \frac{k_\mu k_\nu}{\Lambda^2}\,,
	\label{eq:metric_ds}
\ee
its isometry generators in the vertical space are 
\begin{align}\label{eq:generators_desitter_curved}
	\mathcal{T}_S^\lambda = \sqrt{1 +\frac{\bar k^2}{\Lambda^2} }\ \frac{\partial}{\partial k_\lambda}\,,   \qquad \mathcal{J}^{\mu \nu} = k_\rho(\delta^\nu_\lambda a^{\mu\rho}-\delta^\mu_\lambda a^{\nu\rho}) \frac{\partial}{\partial k_\lambda}\,,
\end{align}
satisfying 
\be
[\mathcal{T}^\alpha_S, \mathcal{T}^\beta_S] = \frac{\mathcal{J}^{\alpha\beta}}{\Lambda^2}\,, \quad\quad [\mathcal{T}^\alpha_S, \mathcal{J}^{\beta\gamma}]= a^{\alpha\beta} \mathcal{T}^\gamma_S - a^{\alpha\gamma} \mathcal{T}^\beta_S\,, \quad\quad [\mathcal{J}^{\alpha\beta},\mathcal{J}^{\gamma\delta}]=a^{\beta\gamma}\mathcal{J}^{\alpha\delta} - a^{\alpha\gamma}\mathcal{J}^{\beta\delta} - a^{\beta\delta}\mathcal{J}^{\alpha\gamma} + a^{\alpha\delta}\mathcal{J}^{\beta\gamma}\,.
\label{cov_generators_curved}
\ee 
This choice of the generators of translations $\mathcal{T}^\alpha_S$ lead to a finite transformation that can be identified with a composition law of the momenta~\cite{Pfeifer:2021tas}
\be
(p\oplus q)_\mu \,= \,p_\mu \left(\sqrt{1+ q_\rho q_\nu a^{\rho\nu}}+\frac{p_\rho q_\nu a^{\rho\nu}}{ \Lambda^2\left(1+\sqrt{1+ p_\rho p_\nu a^{\rho\nu}/\Lambda^2}\right)}  \right) + q_\mu \,,
\label{DCLSnyder-1}
\ee 
corresponding  to Snyder kinematics in the Maggiore representation~\cite{Battisti:2010sr} when $a^{\mu\nu}\to\eta^{\mu\nu}$. That is why we used the subindex $S$ in the previous equations.

Different kinematics can be associated to the same momentum metric: they share the same Hamiltonian (squared geometric distance in momentum space) but have different composition laws (finite translations). We now discuss possible different choices of translations, so new kinematics are obtained. We also show how the generators of translations are related to a noncommutativity of spacetime.

\subsection{\texorpdfstring{$\kappa$}{k}-Poincaré kinematics}
$\kappa$-Poincaré algebra~\cite{Majid1994} can be obtained from the following choice of translation generators
\be
\mathcal{T}^\mu_\kappa= \mathcal{T}^\mu_S+Z_\alpha \frac{\mathcal{J}^{\mu\alpha}}{\Lambda}\,,
\label{cov_gen_kappa2}
\ee
being $Z_\mu=n_\alpha {e^\alpha}_\mu$ and  $\mathcal{J}^{\mu\alpha}$ the same generators of Eq.~\eqref{eq:generators_desitter_curved}, leading to the algebra
\be
[\mathcal{T}^\alpha_\kappa, \mathcal{T}^\beta_\kappa] =\frac{Z_\gamma}{\Lambda}\left(\mathcal{T}^\alpha_\kappa a^{\beta\gamma} -\mathcal{T}^\beta_\kappa a^{\alpha\gamma} \right) \,, \quad\quad [\mathcal{T}^\alpha_S, \mathcal{J}^{\beta\gamma}]= a^{\alpha\beta} \mathcal{T}^\gamma_\kappa - a^{\alpha\gamma} \mathcal{T}^\beta_\kappa+\frac{Z_\delta}{\Lambda}\left(a^{\delta\beta} \mathcal{J}^{\alpha\gamma}-a^{\delta\gamma} \mathcal{J}^{\alpha\beta}\right)\,.
\label{cov_algebra_kappa_curved}
\ee
By replacing $\eta^{\mu\nu}\to a^{\mu\nu}$ and $n_\mu\to Z_\mu$ in the composition law of $\kappa$-Poincaré in the so-called classical basis~\cite{Borowiec2010}
\be
(p\oplus q)_\mu= p_\mu\left(\sqrt{1+\frac{q^2}{\Lambda^2}}+\frac{q\cdot n}{\Lambda}\right)+q_\mu+n_\mu\left(\frac{\sqrt{1+p^2/\Lambda^2}-p\cdot n/\Lambda}{1-((p\cdot n)^2-p^2)/\Lambda^2}\left(q\cdot n+\frac{(p\cdot n) (q\cdot n)-  p\cdot q  }{\Lambda}\right)-q\cdot n\right) \,,
\ee
we could explicitly compute the corresponding composition law associated to the metric~\eqref{eq:metric_ds}, which is
\be
(p\oplus q)_\mu= p_\mu\left(\sqrt{1+\frac{\bar q^2}{\Lambda^2}}+\frac{\bar q\cdot \bar Z}{\Lambda}\right)+q_\mu+Z_\mu\left(\frac{\sqrt{1+\bar p^2/\Lambda^2}-\bar p\cdot \bar Z/\Lambda}{1-(\bar p\cdot \bar Z-\bar p^2)/\Lambda^2}\left(\bar q\cdot \bar Z+\frac{(\bar p\cdot \bar Z) (\bar q\cdot \bar Z)-\bar p\cdot \bar q}{\Lambda}\right)-\bar q\cdot \bar Z\right) \,,
\ee  
where we have used the notation $\bar k\cdot \bar Z=k_\mu Z_\nu a^{\mu\nu}$. This composition law has not been previously calculated in the literature.  Note that this composition law is associative, which is not the case of Snyder kinematics, as can be easily checked from Eq.~\eqref{DCLSnyder-1}.

Moreover, the composition law of $\kappa$-Poincaré presents some crucial differences with respect to the one of Snyder kinematics. First, this composition law is not invariant under spatial isometry transformations, because $Z_\mu$ is not invariant under such transformations (isometries does not leave invariant the form of the tetrad in general). Second, this composition law explicitly depends on the tetrad describing the curvature on spacetime, i.e., it depends on the choice of frame used for defining observers. This means that different tetrads lead to different momentum conservation laws and, thus, different total momentum of a system of particles for different observers. This goes against the independency of the results of physical measurements for different observers. Both problems are not present for  Snyder kinematics. Then, even if in that case the associativity property is lost, Snyder kinematics seems to be privileged from a geometrical point of view (see~\cite{Relancio:2024axb,Relancio:2024loc} for a similar conclusion from a field theory perspective).

\subsection{Noncommutative spacetime from geometry}

There are two ways to consider a non-commutativity of spacetime from a curved momentum space: either to identify the non-commutative space-time coordinates with the generators of the translations in the momentum space,
\be
\mathcal{X}^\alpha=\mathcal{T}^\alpha\,,
\ee
or with the tetrad of the cotangent bundle metric
\be
\mathcal{X}^\alpha={\tilde e^\alpha}_{\,\,\,\,\mu}(x,k)\frac{\partial}{\partial p_\mu}\,,
\ee
satisfying 
\be
{\tilde  e^\alpha}_{\,\,\,\,\mu}(x,k) \eta_{\alpha\beta}{\tilde  e^\beta}_{\,\,\,\,\nu}(x,k)=g_{\mu\nu}(x,k)\,.
\ee
Both approaches lead to a noncommutativity, but different from each other. While from an algebraic point of view the former possibility is the one usually considered (see the discussion of~\cite{Carmona:2021gbg}, and references therein), from a geometrical point of view the latter seems more natural ~\cite{Relancio:2021ahm}.

While for Snyder models one finds that the noncommutativity constructed from the tetrad of the metric is not the usual Snyder noncommutativity derived from the isometries generators (Eq.~\eqref{cov_generators_curved}) but something of the form 
\be
[\mathcal{X}^\alpha_S, \mathcal{X}^\beta_S] = f(\bar k^2)\frac{\mathcal{J}^{\alpha\beta}}{\Lambda^2}\,,
\ee
the same algebra (up to a global sign) is found for $\kappa$-Poincaré  when regarding the construction of the spacetime noncommutativity for both ways. Note that in both cases, the noncommutativity in $\kappa$-Poincaré depends explicitly on the tetrad of spacetime, so it is not invariant for different choices of the tetrad. This is not the case of Snyder noncommutativity, which depends solely on the spacetime metric.  Then, as in the previous subsection, Snyder kinematics seems to be privileged over $\kappa$-Poincaré from this perspective.

\section{Conclusions}
\label{conclusions}
In this paper we show a simple way to construct the geometrical ingredients of a generalized Hamilton space starting from a given metric depending on all the phase-space variables: the Hamiltonian can be defined as the squared distance in momentum space (which can be done since the vertical affine connection can be easily computed for a given metric); the nonlinear coefficients can be obtained from this Hamiltonian as it is done in Hamilton spaces; and the remaining ingredients can be constructed with these objects. This allows one to work with these spaces, which have a richer structure (in the vertical space) than Hamilton spaces. 

We show several properties of cotangent bundle geometries when the Hamiltonian is autoparallel, and for the particular case of being a homogeneous function of the momenta or an arbitrary function of a homogeneous function of momenta. Moreover, we discuss what kind of momentum dependency is allowed in cotangent bundle metrics such that they possess the same number of isometries of spacetime as the corresponding momentum independent metric (this last one obtained by taking the limit of the momentum going to zero in the  general metric). This restricts the possible momentum terms, leaving only space to those preserving linear Lorentz invariance in momentum space. 

We also discuss the isometries of momentum space. From translations in momentum space, which can be identified with a deformed conservation law of momenta in some quantum gravity models, we find that $\kappa$-Poincaré kinematics depends explicitly on the choice of the tetrad used. This seems to imply that Snyder kinematics is privileged, since it is the spacetime metric and not the explicit form of the tetrad that appears in the composition law.

A comparison with other papers considering a deformed algebra of  anti-de Sitter spacetime, such as~\cite{Ballesteros2019}, cannot be made, since our geometrical approach does not allow for the deformed Lorentz generators appearing in these papers to be momentum isometries of the cotangent bundle metric. Therefore, in order to consider such an algebra, the restriction of the Hamiltonian to be autoparallel should be lifted, which goes beyond the scope of this work.

To sum up, the study carried out in this paper is important for several reasons. On the one hand, since we show how to obtain all the necessary ingredients  for describing the geometrical structure of generalized Hamilton spaces, the formal developments obtained here paves the way to consider them at the formal mathematical level. On the other hand, they will also be important in the various applications mentioned at the beginning of the introduction, and, in particular, in quantum gravity.  When imposing a consistent lift of symmetries from flat to curved spacetimes, we find that some of the possible basis are restricted and, even more, some kinematical models seems to be privileged. This restricts the possible kinematical deformations of SR within the framework of DSR even at the flat spacetime scenario, since it should be recovered as a smooth limit of the curved one. In addition, some of the phenomenological consequences of DSR, such as time delays of massless particles with different energies~\cite{Addazi:2021xuf,AlvesBatista:2023wqm}, depend on the basis and the model, so the results of this paper should be taken into account for such studies.


\section*{Acknowledgments}
This research is supported by the Spanish Ministry of Science and Innovation MICIN, as well as by the Regional Government of Castilla y León (Junta de Castilla y León, Spain), with funding from European Union NextGenerationEU   through the QCAYLE project (PRTRC17.I1).  
The authors would also like to thank C. Pfeifer for fruitful discussions.

\bibliography{QuGraPhenoBib}

\end{document}